\documentclass[a4paper,UKenglish,cleveref, autoref,  thm-restate]{lipics-v2021}
\pdfoutput=1 %uncomment to ensure pdflatex processing (mandatatory e.g. to submit to arXiv)
\hideLIPIcs  %uncomment to remove references to LIPIcs series (logo, DOI, ...), e.g. when preparing a pre-final version to be uploaded to arXiv or another public repository

\allowdisplaybreaks
\usepackage{amsfonts,amsmath,amssymb,amsthm}
\usepackage{cite,enumerate,hyperref}
\usepackage[colorinlistoftodos,textsize=small,color=red!25!white,obeyFinal]{todonotes}
\usepackage{csquotes}

\usepackage{xcolor,xspace}

\usepackage{algorithm}
\usepackage{algorithmic}

\newcommand{\cC}{\mathcal{C}}

\newcommand{\cd}{\text{ :- }}

\newcommand{\Datalog}{\text{\sf Datalog}}

\newcommand{\zero}{\mathcal{O}}

\usepackage{bm} % bold math fonts, e.g. $\bm \Sigma$ will give a bold-face $\Sigma$
\usepackage{url}
\usepackage{comment}

\usepackage[colorinlistoftodos]{todonotes}

\usepackage{enumitem}
\usepackage{subcaption}

\usepackage{microtype}

\newcommand{\set}[1]{\{#1\}}                    % Set (as in \set{1,2,3}).
\newcommand{\bag}[1]{\{\hspace{-1mm}\{#1\}\hspace{-1mm}\}}                    % bag (as in \bag{1,2,3}).
\newcommand{\bagof}[2]{\{\hspace{-1mm}\{{#1}\mid{#2}\}\hspace{-1mm}\}}        % Set (as in \setof{x}{x>0}).
\newcommand{\calB}{\mathcal B}

\newcommand{\calW}{\mathcal W}

\newtheorem{exnum}{Example}%[section]
\newtheorem{ex}[exnum]{Example}

\newcommand{\defeq}{\stackrel{\text{def}}{=}}

\newcommand{\R}{\mathbb R} % the real numbers
\newcommand{\Name}{\text{\sf Datalog}^\circ}
\newcommand{\trop}{\text{\sf Trop}}
\newcommand{\true}{\text{\sf true}}
\newcommand{\false}{\text{\sf false}}

\newcommand{\adom}{\textsf{ADom}}

\author{Sungjin Im}{University of California Merced    }{}{}{}

\author{Benjamin Moseley}{Carnegie Mellon University     }{}{}{}
\author{Hung Ngo}{RelationalAi   }{}{}{}
\author{Kirk Pruhs}{University of Pittsburgh  }{}{}{}

\authorrunning{Sungjin Im, Benjamin Moseley, Hung Ngo and Kirk Pruhs}

\Copyright{Sungjin Im, Benjamin Moseley, Hung Ngo and Kirk Pruhs}



\title{On the Convergence Rate of Linear $\Name$ over Stable Semirings}%\thanks{Moseley was supported in part by a Google Research Award, an Inform Research Award, a Carnegie Bosch Junior Faculty Chair, and NSF grants CCF-2121744 and  CCF-1845146. Pruhs was supported in part by NSF grants  CCF-1907673,  CCF-2036077, CCF-2209654 and an IBM Faculty Award. Im was supported in part by NSF grants CCF-1844939   and CCF-2121745.}} %TODO Please add

\funding{Moseley was supported in part by a Google Research Award, an Inform Research Award, a Carnegie Bosch Junior Faculty Chair, and NSF grants CCF-2121744 and  CCF-1845146. Pruhs was supported in part by NSF grants  CCF-1907673,  CCF-2036077, CCF-2209654 and an IBM Faculty Award. Im was supported in part by NSF grants CCF-1844939  
 and CCF-2121745.}

\nolinenumbers %uncomment to disable line numbering

\EventEditors{Graham Cormode and Michael Shekelyan}
\EventNoEds{2}
\EventLongTitle{27th International Conference on Database Theory (ICDT 2024)}
\EventShortTitle{ICDT 2024}
\EventAcronym{ICDT}
\EventYear{2024}
\EventDate{March 25--28, 2024}
\EventLocation{Paestum, Italy}
\EventLogo{}
\SeriesVolume{290}
\ArticleNo{11}
\begin{document}

\maketitle

\begin{abstract}
$\Name$ is an extension of $\Datalog$, where  instead of a program being a collection of
union of conjunctive queries over the standard Boolean semiring, a program may now be a
collection of sum-product queries over an arbitrary commutative partially ordered
pre-semiring. $\Name$ is more powerful than $\Datalog$ in that its additional algebraic
structure alows for supporting recursion with aggregation. At the same time, $\Name$ retains
the syntactic and semantic simplicity of $\Datalog$: $\Name$ has declarative least fixpoint
semantics. The least fixpoint can be found via the na\"ive evaluation algorithm that
repeatedly applies the immediate consequence operator until no further change is possible.

It was shown in~\cite{Khamis0PSW22} that, when the underlying semiring is $p$-stable,
then the na\"ive evaluation of any $\Name$ program over the semiring converges in a finite
number of steps. However, the upper bounds on the rate of convergence  were exponential in
the number $n$ of ground IDB atoms.

This paper establishes polynomial upper bounds on the convergence rate  of the na\"ive
algorithm on {\bf linear} $\Name$ programs, which is quite common in practice. In
particular, the main result of this paper is that the convergence rate of linear $\Name$
programs under  any $p$-stable semiring is  $O(pn^3)$.  Next, we study
the convergence rate in terms of the number of elements in the semiring for linear $\Name$
programs.  When $L$ is the number of elements, we show that the convergence rate is bounded by $O(pn \log
L)$.  This significantly improves the convergence rate for small $L$.  
\end{abstract}

\section{Introduction}

In order to express common recursive computations with aggregates in modern data analytics
while retaining the syntactic and semantic simplicity of $\Datalog$, \cite{Khamis0PSW22}
introduced $\Name$, an extension of $\Datalog$ that allows for aggregation and recursion
over an arbitrary {\em commutative partially ordered pre-semiring} (POPS).  $\Datalog$ is
exactly $\Name$ over the Boolean semiring.
Like $\Datalog$, $\Name$ has a declarative least fixpoint semantics, and the least fixpoint
can be found via the na\"ive iteration algorithm that repeatedly applies the immediate
consequence operator until no further change is possible. Moreover, its additional algebraic
structure allows for common recursions with aggregations.

However unlike $\Datalog$, the na\"ive evaluation of
a  $\Name$ program may not always converge in a finite number of steps. 
The convergence of a $\Name$ program over a given POPS can be studied through its ``core semiring'', which is
where we focus our attention on in this paper. This paper will only study $\Name$ programs
over commutative semirings, referring the readers to~\cite{Khamis0PSW22} for the
generality of POPS.

It is known that the commutative semirings for which the iterative evaluation of  $\Name$ programs is guaranteed to converge
are exactly those semirings that are stable~\cite{Khamis0PSW22}. 
A semiring is $p$-stable if the  number of iterations required for any one-variable recursive linear $\Name$ program to reach 
a fixed point is at most $p$, and a semiring is stable if there exists a $p$ for which it is $p$-stable.  
Further, every non-stable semiring has a simple (linear) $\Name$ program with one variable with the property that
the iterative evaluation of this program over that semiring will not converge. 
Thus it is natural to concentrate on $\Name$ programs over stable semirings. 
Previously, the best known upper bound on the convergence rate, which is the number of iterations
until convergence,  is $\sum_{i=1}^n (p+2)^i = \Theta(p^n)$ steps, 
where  $n$ is the number of ground atoms for the IDB's that ever have
a nonzero value at some point in the iterative evaluation of the $\Name$ program, and the underlying semiring is $p$-stable.
In contrast there are no known lower bounds that show that  iterative evaluation requires an exponential (in the parameter $n$) number of steps to reach convergence.

The exact general upper bound on the convergence rate of $\Name$ programs over $p$-stable
semirings is open, even for the special case of linear $\Name$ programs.
Linear $\Name$ programs are quite common in practice.\footnote{For example, we can express
transitive closure, all-pairs-shortest-paths, or weakly connected components in linear $\Name$.}
Thus, in this paper we focus on this ``easiest'' case where the exact convergence rate
is not known.

Currently, the best known upper bound on the convergence rate of linear $\Name$ programs
over $p$-stable semirings is $\sum_{i=1}^n (p+1)^i$ steps.  
This bound is unsatisfactory in the following sense. The
prototypical example of a $p$-stable semiring is the tropical semiring $\trop_p^+$, (see
Section~\ref{sect:prelims} for a definition of $\trop_p^+$); in this case, it is known that
the na\"ive algorithm converges in $O(pn)$ steps for linear $\Name$
programs~\cite{Khamis0PSW22}. These results leave open the possibility that  the convergence rate of the na\"ive algorithm on linear $\Name$ programs over $p$-stable semirings could be exponentially smaller than the best known guarantee.

This paper seeks to  obtain tighter bounds
on the convergence rate  of na\"ive evaluation of {\bf linear} $\Name$ programs.
As the iterative evaluation of $\Name$ programs is a reasonably natural and important
algorithm/process, bounding the running time of this process 
 is of both theoretical interests
and practical interests. In practice, a known upper bound on the convergence rate
allows the database system to determine {\em before hand} an upper bound on the
number of iterations that will be required to evaluation a particular $\Name$ program.

\subsection{Background}

Before stating our main results, we need to back up to set the stage a bit.
A (traditional) $\Datalog$
program $\Pi$ consists of a set of rules of the form:
\begin{align}
  R_0(\bm X_0) &\cd R_1(\bm X_1) \wedge \cdots \wedge R_m(\bm X_m) \label{eq:datalog}
\end{align}
where $R_0, \ldots, R_m$ are predicate names (not necessarily distinct) and each $\bm X_i$
is a tuple of variables and/or constants. The atom $R_0(\bm X_0)$ is called the head, and
the conjunction $R_1(\bm X_1) \wedge \cdots \wedge R_m(\bm X_m)$ is called the body of the rule.
Multiple rules with the same head are interpreted as a disjunction. A predicate that
occurs in the head of some rule in $\Pi$ is called an {\em intensional database predicate}
or IDB, otherwise it is called an {\em extensional database predicate} or EDB.  The EDBs
form the input database,  and the IDBs represent the output instance computed by $\Pi$. The
finite set of all constants occurring in an EDB is called the {\em active domain}, and
denoted $\adom$.
A $\Datalog$ program is {\em linear} if every rule has at most one IDB predicate
in the body.

There is an implicit existential quantifier over the body for all variables
that appear in the body but not in the head, where the domain of the existential quantifier
is $\adom$. 
In a linear Datalog program every conjunction has at most one IDB.

\begin{ex} \label{ex:TC}
The textbook example of a linear $\Datalog$ program is the following one, which computes the
transitive closure of a directed graph, defined by the edge relation $E(X,Y)$:
\begin{align*}
  T(X,Y) &\cd E(X,Y) \\
  T(X,Y) &\cd T(X,Z) \wedge E(Z,Y)
\end{align*}
Here $E$ is an EDB predicate, $T$ is an IDB predicate, and $\adom$ consists of the vertices in the graph.
The other way to write this program is to write it as a union of conjunctive queries
(UCQs),
where the quantifications are explicit:
\begin{align}
  T(X,Y) & \cd E(X,Y) \vee \exists_Z \ (T(X,Z) \wedge E(Z,Y)) \label{eq:datalog:intro}
\end{align}
\end{ex}

A $\Datalog$ program can be thought of as a function (called the {\em immediate consequence
operator}, or {\em ICO}) that maps a set of ground IDB atoms to a set of ground IDB atoms. Every
rule in the program is an inference rule that can be used to infer new ground IDB atoms from
old ones. For a particular EDB instance, this function has a unique least fixpoint which
can be obtained via repeatedly applying the ICO until a fixpoint is
reached~\cite{DBLP:books/aw/AbiteboulHV95}.
This least fixpoint is the semantics of the given $\Datalog$ program.
The algorithm is called the {\em na\"ive evaluation algorithm}, which converges
in a polynomial number of steps in the size of the input database, given that
the program is of fixed size.

Like $\Datalog$ programs, a $\Name$
program consists of a set of rules, where the unions of conjunctive queries are now
replaced with sum-product queries over a commutative semiring $\bm S = (S, \oplus, \otimes, 0, 1)$; see Section~\ref{sect:prelims}.
Each rule has the form:
\begin{align} R_0(X_0) &\cd \bigoplus R_1(\bm X_1) \otimes
    \cdots \otimes R_m(\bm X_m) \label{eq:t:monomial}
  \end{align}
  where   sum ranges over the active $\adom$ of the variables not in $X_0$. 
Further each ground atom in an EDB predicate or IDB predicate is associated with an element of the semiring $\bm S$, and the element associated with a tuple in an EDB predicate is specified in the input.
The EDBs
form the input database,  and the ground atoms in the IDB's that have an associated semiring value that is nonzero  represent the output instance computed by the $\Name$ program. 
Note that in a $\Datalog$ program the ground atom present in the input or output databases can be thought of as
those that are  associated with
the element $1$ in the standard Boolean semiring. 
A $\Datalog$ program is a $\Name$ program over the Boolean semiring
$\bm S = (\{\true,\false\}, \vee, \wedge, \false, \true)$.
Again a  $\Name$ program is  linear if every rule has no more than one IDB prediciate in its body.

\begin{ex} \label{ex:APSP} A simple example of a linear $\Name$ program is the following,
  \begin{align}
    T(X,Y) &\cd E(X,Y) \oplus \bigoplus_Z T(X,Z) \otimes E(Z,Y), \label{eqn:linear:tc}
  \end{align}
which is~\eqref{eq:datalog:intro} with $(\vee, \wedge, \exists_Z)$
replaced by $(\oplus,\otimes, \bigoplus_Z)$.

  When interpreted over the Boolean semiring, we obtain the transitive closure program from
 Example~\ref{ex:TC}. When interpreted over the {\em tropical semiring} $\trop^+ = (\R_+ \cup
 \{\infty\}, \min, +, \infty, 0)$, we have  the {\em All-Pairs-Shortest-Path} (APSP)
 program, which  computes the shortest path length $T(X,Y)$ between all pairs $X,Y$ of
 vertices in a directed graph specified by an edge relation $E(X,Y)$, where  the semiring
 element associated with $E(X,Y)$ is the length of the directed edge $(X,Y)$ in the graph:
  \begin{align}
    T(X,Y) &\cd \min\left(E(X,Y), \min_Z (T(X,Z) + E(Z,Y))\right) \label{eqn:apsp}
  \end{align}
\end{ex}

A $\Name$ program can be thought of as an immediate consequence
operator (ICO). To understand how the ICO works in $\Name$,
consider a rule with head $R$ and let $A$ be a ground
atom for $R$ with associated semiring value $y$, and assume that for $A$ the body of this rule
evaluates to the semiring element $x$. 
As a result of this,  the ICO associates $A$ with with $x\oplus y$.
Note that the functioning of the $\Name$ ICO, when the semiring is  the standard Boolean semiring,
is identical to how  the $\Datalog$ ICO functions. 
The iterative evaluation of a $\Name$ program works in rounds/steps, where initially the
semiring element 0 is associated with each possible ground atom in an IDB,
and on each round
the ICO is applied to the current state. 
So in the context of the $\Name$ program in Example \ref{ex:TC}, if $(a, b)$ was a ground atom 
in $T$ with associated semiring element $x $, meaning that the shortest known directed path from vertex $a$ to vertex $b$
has length $x$,
and $(b, c)$ was a ground atom in $E$ with associated semiring element $y $, meaning that there is a directed edge from $b$ to $c$ with length $y$,
then the ICO would make the semiring element associated with $A$ the minimum (as minimum is the addition operation in the tropical semiring) of its current value and  $x+y$ (as normal additional is the  multiplication operation in the  tropical semiring).

Since the final associated semiring values of the ground atoms in an IDB  are not initially known, it is natural to think of them
as (IDB)  variables.  
Then the grounded version of the ICO   of a $\Name$program is a map $\bm f : S^n \rightarrow S^n$, 
where  $n$ is the number of ground atoms for the IDB's that ever have
a nonzero value at some point in the iterative evaluation of the $\Name$ program. 
For instance, in \eqref{eqn:apsp}, 
there would be one variable  for each pair $(x, y)$ of vertices where there is a directed path from $x$ to $y$ in the graph.  
So the grounded version of the ICO of a $\Name$ program has the following form:%
\begin{align}
  X_1 \cd & f_1(X_1, \ldots, X_n) \nonumber\\
          & \ldots \label{eq:grounded:pi} \\
  X_n \cd & f_n(X_1, \ldots, X_n) \nonumber
\end{align}
where the $X_i$'s are the  IDB variables, and $f_i$ is the component of $\bm f$ corresponding to the IDB variable $X_i$. 
Note that each component function $ f_i$ is a multivariate polynomial in the IBD variables
of degree at most the maximum 
number of terms in any product in the body of some rule in the $\Name$ program. 
After
$q$ iterations of the iterative evaluation of a $\Name$ program, the semiring value  associated with the ground atom corresponding to $X_i$ 
will be:
\begin{align}
{f}_i^{(q)}(\bm 0)
\end{align}

\begin{definition}[$p$-stability]
Given a semiring $\bm S = (S, \oplus, \otimes, 0, 1)$ and $u \in S$, let $u^{(p)} := 1
\oplus u \oplus u^2 \oplus \cdots \oplus u^{p}$, where $u^{i} := u \otimes u \otimes \cdots
\otimes u$ ($i$ times).  Then $u\in S$ is {\em $p$-stable} if $u^{(p)}=u^{(p+1)}$, and
semiring $\bm S$ is {\em $p$-stable} if every element $u \in S$ is $p$-stable.
\end{definition}

\begin{definition}[Stability index and convergence rate]
    \label{def:stability-convergence}
A function $f : S^n \to S^n$ is $p$-stable if $f^{(p+1)}(\bm 0) = f^{(p)}(\bm 0)$, where
$f^{(k)}$ is the $k$-fold composition of $f$ with itself.  The {\em stability index} of
$f$ is the smallest $p$ such that $f$ is $p$-stable.
The {\em convergence rate} of a $\Name$ program is the stability index of its ICO.
\end{definition}

The following bounds on the rate of convergence of a general (multivariate) polynomial
function $f : S^n \to S^n$, where $\bm S $ is $p$-stable; this result naturally infers bounds on
the convergence rate of $\Name$ programs over $p$-stable semirings.

\begin{theorem}[\cite{Khamis0PSW22}] \label{th:Hungmain:intro}
The convergence rate of a $\Name$ program over  a $p$-stable commutative semiring is at most $\sum_{i=1}^{n}(p+2)^i$. Further, the convergence rate is at most $\sum_{i=1}^n (p+1)^i$ if the $\Name$ program is linear. Finally, 
the convergence rate of a $\Name$ program is at most $n$ if $p = 0$.
\end{theorem}

Thus the natural question left open by ~\cite{Khamis0PSW22} was whether these upper bounds
on the rate of convergence are (approximately) tight, and thus convergence can be exponential,
or whether significantly better bounds are achievable.

\subsection{Our Results}

When a $\Name$ program over a semiring $\bm S$ is {\em linear}, its ICO
$f: S^n \rightarrow S^n$ is a {\em linear} function of the form
$f(x) = A \otimes x \oplus b$,
where $A$ is an $n$ by $n$ matrix with entries from $S$,
$b$ is an
$n \times 1$ column vector  with entries from $S$, and the scalar multiplication and addition
are from $\bm S$.
To simplify notations, we will use $+$ and $\cdot$ to denote the operations $\oplus$ and
$\otimes$ respectively. Further, following the convention, we may omit $\otimes$ if it is clear from the context. Then, we have
$$f(x) = A x + b$$

\begin{ex} \label{ex:APSPMatrix}
For the APSP $\Name$ program, $n$ is the number of edges in the graph,
$b_{(u, v)}$ is $E(u, v)$,  and the matrix $A$ would be:
\begin{equation}
A_{(u,v), (u, w)}=
    \begin{cases}
        E(v, w) & \text{if } u \ne v\\
        0 & \text{otherwise}
    \end{cases}
\end{equation}
\end{ex}

The stability index of $f$ can easily be expressed in terms of the matrix $A$ and vector
$b$. 
Letting
$A^0 = I$ where $I$ is the identity matrix, we have
\begin{align*}
f^{(k)}(x) = A^k x + A^{(k-1)}b &&
\text{where }
A^{(k)} := \sum_{h=0}^{k-1} A^h.
\end{align*}
Thus, the linear function $f=Ax+b$ is $p$-stable if and only if $A^{(p)}b = A^{(p+1)}b$. Using the simple form of $f$ for the linear case, we can rewrite Definition~\ref{def:stability-convergence} into the following simpler form.
\begin{observation}[Convergence Rate of a Linear $\Name$ Program]
    \label{defn:convergencenaive} ~
    \begin{itemize}
        \item
The convergence rate for a particular linear $\Name$ program, with associated matrix $A$  and vector $b$,
is
the minimum natural number $p$ such that
$A^{(p)}b = A^{(p+1)}b$.
\item
The convergence rate for general linear $\Name$ programs
over a semiring $S$ is then the maximum
over all choices of $A$ and $b$, of the convergence rate for that particular $A$ and $b$.
    \end{itemize}
\end{observation}

Our first  result is a bound of $O(p n^3)$ on the  rate
of convergence for  linear $\Name$ programs.

\begin{theorem} \label{thm:pn3upper} Every linear $\Name$ program over a $p$-stable
commutative semiring $\bm S$ converges in $O(p n^3)$ steps, where $n$ is the total number of
ground IDB atoms.
\end{theorem}

The proof Theorem \ref{thm:pn3upper} can be found in Section \ref{sect:pncubedupper}, but we
will give a brief overview of the main ideas of the proof here.
Consider the complete $n$-vertex loop-digraph $G$ where the edge from $i$ to $j$ is labeled
 with entry $A_{i,j}$. Then note that row $i$ column $j$ of $A^{(h)}$  is the sum -- over
 all walks $W$ from $i$ to $j$ of length $\leq h$ -- of the product of the edge labels on
 the walk. We show that the summand corresponding to a walk $W$ with more than $h=\Omega(p
 n^3)$ hops doesn't change the sum   $A^{(h)}_{i,j}$. We accomplish this by rewriting  the
 summand corresponding to $W$, using the commutativity of multiplication in $S$, as the
 product of a simple path and of multiple copies of at most $n^2 -n$ distinct closed walks.
 We then note that, by the pigeon hole principle, one of these closed walks, say $C$, must
 have mulitplicity greater than $p$. We then conclude that the summand corresponding to $W$
 will not change the sum $A^{(h)}_{i,j}$ by appealing to the stability of the semiring
 element that is the product of the edges in $C$.

In the conference version of this paper~\cite{icdt/ImM0P24}, and the original
arxiv version of the paper, we claimed
a matching lower bound. That is we claimed that
for any $p, n \geq 1$, there is a linear $\Name$ program over a $p$-stable semiring $\bm S$ that requires 
$\Omega(p n^3)$ steps to converge.
It was subsequently pointed out that our claim was incorrect~\cite{ParisPersonal}. 
Still in Section  \ref{sec:pn3lower} we briefly discuss our lower bound attempt,
and the underlying conceptual error 
that led to our mistake, as we believe that it is instructive. 

A natural question is whether $\Name$ programs over semirings with small
sets will converge more quickly. In section \ref{sect:Lupper} we answer this
question in the affirmative by showing that the rate of convergence of a $\Name$ program
over a $p$-stable commutative semiring with $L$ elements is $O( p n  \log L)$.

\begin{theorem}
\label{thm:plupper}
Every linear $\Name$ program over a $p$-stable commutative semiring that contains $L$
elements in its ground set converges
in $O(p n \log L)$ steps.
\end{theorem}

Let us explain the high level idea of the proof, with some simplifying assumptions, starting
with the assumption that the stability of $\bm S$ is $p = 1$. Again we think of
$A^{(k)}_{i,j}$ as the sum of products over walks $W$ from $i$ to $j$ with at most $k$ hops.
Now consider a walk $W$ from $i$ to $j$ consisting of $\Omega(n \log L)$ hops. Since there
are $n$ vertices, there exists a vertex $v$ such that there are at least $\Omega(\log L)$
prefixes $P_1, P_2, \ldots, P_q$ of $W$ ending at  $v$. Let $C_i$ be the closed walk
starting and ending at $v$ such that appending $C_i$ to $P_i$ produces $P_{i+1}$. Let us
make the simplifying assumption that all of these closed walks $C_i$ are distinct. The key
observation is that if we delete any subset $\mathcal C$ of these $\Omega(\log L)$ closed
walks from $W$,  the result will still be a  walk from $i$ to $j$. Thus there are at least
$2^{\Omega(\log L)}$  walks from $i$ to $j$ that can be formed by deleting a subset
$\mathcal C$ of these closed walk. Since there are at most $L$ distinct elements, by the
pigeon hole principle, there must be some element $e$ of $S$ such that that are two subsets
$\mathcal C, \mathcal C'$ where the product of the edges in them are the same. We then
conclude by the stability of $e$ that the summand corresponding to $W$ does not change the
sum $A^{(k)}_{i,j}$.

In the conference version of this paper~\cite{icdt/ImM0P24}, and the original
arxiv version of the paper, we claimed
a nearly matching lower bound. We now believe that this lower bound
is also not correct, for reasons roughly similar to those that led to our
incorrect $\Omega(p n^3)$ lower bound.

Finally in the appendix we consider a special case from~\cite{Khamis0PSW22}, in which the
commutative semiring $\bm S$ is naturally ordered.

\section{Preliminaries}
\label{sect:prelims}

In this section, we introduce the notation and terminology used throughout the paper.

\begin{definition}[Semiring]
A {\em semiring}~\cite{semiring_book} is a tuple
$\bm S = (S, \oplus, \otimes, 0, 1)$ where
\begin{itemize}
    \item $\oplus$ and $\otimes$ are
binary operators on $S$,
\item  $(S, \oplus, 0)$ is a commutative
monoid, meaning $\oplus$ is commutative and associative, and $0$ is the identity for $\oplus$,
\item  $(S, \otimes, 1)$ is a monoid, meaning $\otimes$ is associative, and $0$ is the identity for $\oplus$,
\item $0$ annilates every element $a \in S$, that is $a\otimes 0 = 0 \otimes a = 0$, and
\item $\otimes$ distributes over
$\oplus$.
\end{itemize}
A commutative semiring $\bm S = (S, \oplus, \otimes, 0, 1)$ is a semiring where
additionally $\otimes$ is commutative.
\end{definition}

If $A$ is a set and $p \geq 0$ a natural number, then we denote by $\calB_p(A)$ the set of
bags (multiset) of $A$ of size $p$, and
$\calB_{\texttt{fin}}(A) := \bigcup_{p\geq 0}\calB_p(A)$.
We denote bags as in $\bag{a,a,a,b,c,c}$.
Given $\bm x,\bm y \in \calB_{fin}(\R_+ \cup \infty)$, define
\begin{align*}
  \bm x \uplus \bm y &:= \mbox{bag union of $\bm x,\bm y$} &
  \bm x + \bm y &:= \bagof{u+v}{u \in \bm x, v \in \bm y}
\end{align*}

\begin{ex} \label{ex:trop:p}
    For any multiset
    $\bm x = \bag{x_0, x_1, \ldots, x_n}$, where
    $x_0\leq x_1 \leq \ldots \leq x_n$, and any $p \geq 0$, define:
    \begin{align*}
      {\min}_p(\bm x) &:= \bag{x_0, x_1, \ldots, x_{\min(p,n)}}
    \end{align*}
    In other words, $\min_p$ returns the smallest $p+1$ elements of the
    bag $\bm x$.  Then, for any $p \geq 0$, the following is a semiring:
    \begin{align*}
      \trop^+_p &:= (\calB_{p+1}(\R_+\cup\set{\infty}), \oplus_p, \otimes_p, \bm 0_p, \bm 1_p)
    \end{align*}
    where:
    \begin{align*}
      \bm x \oplus_p \bm y \defeq & {\min}_p(\bm x \uplus \bm y) &
      \bm 0_p \defeq & \set{\infty, \infty, \ldots, \infty} \\
      \bm x \otimes_p \bm y \defeq & {\min}_p(\bm x + \bm y) &
      \bm 1_p \defeq & \set{0,\infty, \ldots, \infty}
    \end{align*}
    For example, if $p=2$ then
    $\bag{3,7,9} \oplus_2 \bag{3,7,7} = \bag{3,3,7}$ and
    $\bag{3,7,9} \otimes_2 \bag{3,7,7} = \bag{6,10,10}$.
    The following identities are easily checked, for any two finite bags
    $\bm x, \bm y$:
    \begin{align}
      {\min}_p({\min}_p(\bm x) \uplus {\min}_p(\bm y))= & {\min}_p(\bm x \uplus \bm y)&
      {\min}_p({\min}_p(\bm x) + {\min}_p(\bm y))= & {\min}_p(\bm x + \bm y) \label{eq:minp:identity}
    \end{align}
    Note that $\trop^+_0$ is the natural ``min-plus'' semiring that we used in
    Example~\ref{ex:APSP}.
  \end{ex}

In the following fact about $p$-stable commutative semiring will be useful.

\begin{proposition}
    \label{prop:coef-limit}
    Given a $p$-stable commutative semiring $\bm S = (S, \oplus, \otimes, 0, 1)$, for any $u \in S$, we have $p u = (p+1) u$,
    where $pu$ here is shorthand for $\oplus_{i=1}^p u$.
\end{proposition}
\begin{proof}
This follows directly from the $p$-stability of 1, and the fact that $1^2 = 1$.
\end{proof}

Let us now explain the general procedure for creating a matrix $A$ and a vector $b$ from a
linear $\Name$ program $Q$. Each ground tuple for each IDB predicate $R$ can be viewed as a
variable, and ground tuples of EDB predicates can be viewed as constants. Then a
$\Name$ rule, where the head is IDB predicate $R$, can be converted into a collection of
linear equations, one for each ground tuble of $R$.
Since all rules that share IDB $R$ as a the head can be combined via $\oplus$, we can compactly rewrite the entire set of $\Name$ rules as a collection of linear equations of the following form:
$$T_i \leftarrow \bigoplus_{j = 1}^{n}  (A_{ij} \otimes T_j) \oplus b_i$$
where $T_1, T_2, \ldots, T_n$ are the variables corresponding to ground tuples of IDB predicates.
By using more familiar notation $+$, $\cdot$ in the lieu of $\otimes, \oplus$, we can model a
linear $\Name$ program $Q$ by a linear function
$f: S^n \rightarrow S^n$ of the form:
$$f(x) = Ax +b$$
where $n$ is the total number of ground tuples across all IDB predicates, $x$ is an $n$-dimensional vector with entries from $S$,
$A$ is an $n$ by $n$ matrix with entries from the  $S$, and
$b$ is a dimensional column vector with entries from $S$.
For some more examples of converting linear $\Name$ programs into linear equations
see ~\cite{Khamis0PSW22}.

\section{Upper bounding the Convergence as a Function of the Matrix Dimension
and the Semiring Stability}
\label{sect:pncubedupper}

This section is devoted to proving Theorem~\ref{thm:pn3upper}, the main theorem of the paper. More specifically, we will show the following lemma.  This lemma upper bounds the convergence of a $p$-stable semiring and implies Theorem~\ref{thm:pn3upper}.

\begin{lemma}
\label{lemma:stabilityupper}
Let $A$ be an $n \times n$ matrix over a $p$-stable semiring $S$.
Then $A^{(k+1)} = A^{(k)}$, where $k=n (n^2-n) (p+2) + n-1$.
\end{lemma}

Consider the complete $n$-vertex loop-digraph $G$ where the edge from $i$ to $j$ is labeled with
entry $A_{i,j}$. Then
$$A^{h}_{i,j} = \sum_{W \in {\mathcal W^h_{i,j}}} \Phi(W)$$
where ${\mathcal W^h_{i,j}}$ is the collection of all $h$-hop walks from $i$ to $j$ in $G$,
and
$$\Phi(W) =  \prod_{(a,b)  \in W} A_{a,b}$$
is the product off all the labels on all the directed edges in $W$.
That is, row $i$ column $j$ of $A^h$  is the
sum over all $h$-hop walks $W$ from $i$ to $j$ of the product of the labels on the walk.
Similarly then,
$$A^{(h)}_{i,j} = \sum_{g=0}^h A^g_{i,j} = \sum_{g=0}^h \sum_{W \in {\mathcal W^g_{i,j}}} \Phi(W)$$
That is, row $i$ column $j$ of $A^{(h)} $  the
sum over all  walks $W$ from $i$ to $j$ with at most $h$ hops of the product of the labels on the walk.
Further,
$$A^{(h+1)}_{i,j} =  A^{(h)}_{i,j} + A^{h+1}_{i,j} = A^{(h)}_{i,j} + \sum_{W \in {\mathcal W^{h+1}_{i,j}}} \Phi(W)$$

Our proof technique is to show that, by the $p$-stability of
$S$, it must be the case that for each $W  \in {\mathcal W^{p+1}_{i,j}}$
it is the case that
$$   A^{(k)}_{i,j} +    \Phi(W) = A^{(k)}_{i,j} $$
The proof that $A^{(k)}_{i,j} =  A^{(k+1)}_{i,j}$ then immediately follows by applying this
fact to each $W  \in {\mathcal W^{k+1}_{i,j}}$.

Fix $W=i,\dots,j$ to be an arbitrary walk in  ${\mathcal W^{k+1}_{i,j}}$.
Our next intermediate goal is to rewrite $\Phi(W)$ using the commutativity of multiplication in $S$
as the product of a simple path and at most $n^2 -n$ multicycles. That is,
$$ \Phi(W) = \Phi(P) \prod_{h=1}^\ell \Phi(C_h^{z_h}) $$
where $P$ is a simple path from $i$ to $j$ in $G$, each $C_h$ is a simple cycle in $W$
that is repeated $z_h$ times,
and  $\ell \le n^2-n$.
We accomplish this goal via the following tail-recursive construction.  The recursion
 is passed a collection of edges,  and
a parameter $h$.  Initially, the edges are those in $W$ and $h$  is set to zero.

\noindent{\bf Recursive Construction:} The base case is if $W$ is a simple path.
In the base case the path $P$ is set  to $W$ and $\ell$ is set to $0$.
Otherwise:
\begin{itemize}
\item $h$ is incremented
    \item
$C_h$ is set to be an arbitrary simple cycle in $W$.
\item
Let $z_h$ be the minimum over all edges $e \in C_h$ of the number of times
    that $e$ is traversed in $W$.
    \item
Let $W'$  be the collection of  edges in $W$ except that $z_h$ copies of every edge in $C_h$
are removed.
\item
 The construction then recurses on $W'$ and $h$.
\end{itemize}

In Lemma  \ref{lemma:walkminuscycle} we show that a particular statement about $W$ is invariant through
the recursive construction.
A proof Lemma  \ref{lemma:walkminuscycle}  requires the following lemma. The proof of the following lemma (or at least, the techniques needed for a proof) can be found in most introductory graph theory texts, e.g. \cite{wilson10} Theorem 23.1.

\begin{lemma}~
\label{lemma:Eulerian}
\begin{itemize}
    \item
   A loop digraph $G$ has a Eulerian walk from from a vertex $i$ to a vertex $j$, where $i \ne j$, if
   and only if vertex $i$ has out-degree one greater than its in-degree,
   vertex $j$ has out-degree one less than its in-degree,
   every other vertex has equal in-degree and out-degree, and all of the vertices with nonzero degree belong to a single connected component of the underlying undirected graph.
   \item
A loop digraph has an Eulerian cycle that includes a vertex $i$
if and only if every vertex has equal in-degree and out-degree, vertex $i$ has non-zero in-degree, and all of the vertices with nonzero degree belong to a single connected component of the underlying undirected graph.
   \end{itemize}
\end{lemma}

\begin{lemma}
\label{lemma:walkminuscycle}  Let $W$ be a collection of edges that is passed at some point in the
recursive construction. Let $D_1, \ldots, D_h$ be a partition of the edges of $W$ with the property
that if the edges in $W$ were viewed as undirected, then the connected components would be
$D_1, \ldots, D_h$.
\begin{itemize}
\item
If $i \in D_f$ then $D_f$ is a walk from $i$ to $j$.
\item
If $i \notin D_f$ then $D_f$ is a Eulearian circuit.
\end{itemize}
\end{lemma}

\begin{proof}
 The proof is by induction on the number of
steps of the recursive construction. The statement is obviously true for the initial
walk $W$, which is the base case. Now consider one step of the recursive construction.
Removing copies of a cycle from $W$ does not change the difference between the
in-degree and out-degree of any vertex. Thus by Lemma \ref{lemma:Eulerian} the only issue
we need to consider is vertex $i$ and vertex $j$ possibly ending up in different connected components
of $W'$.  Let $D_f$ be the connected component of $W$ that contains $i$ (which also
contains $j$ by the induction hypothesis),
let $D'_a$ be the connected component of $W'$ that contains $i$, and
let $D'_b$ be the connected component of $W'$ that contains $j$, where $a \ne b$.
Then the walk in $D_f$ from $i$ to $j$ must  cross the cut (in either direction) formed by the vertices in
$D'_a$ an odd number of times.  But the edges in $C_h$
must  cross the cut (in either direction) formed by the vertices in
$D'_a$ an even number of times. Thus there must be an edge in $D_f $ minus $z_h$ copies of $C_h$
that must cross the cut (in either direction) formed by the vertices in
$D'_a$. However, then this is a contradiction to $D'_a$ be a connected component in $W'$.
\end{proof}

We now make a sequence of observations, that will eventually lead us to our
proof of Lemma \ref{lemma:stabilityupper}.

\begin{observation}~
\label{obs:cycledecomposition}
\begin{itemize}
    \item $1 \le \ell \le n^2 - n$.
    \item The recursive construction terminates.
    \item $ \Phi(W) = \Phi(P) \prod_{h=1}^\ell  \Phi(C_h^{z_h}) $.
\end{itemize}
\end{observation}

\begin{proof}
As $W$ contains $k+1=n (n^2-n) (p+2) + n$ edges, then it must contain a simple cycle, which implies
$\ell \ge 1$.
Consider one iteration of our recursive construction.
There must be a directed edge $e \in C_h$ that
    appears exactly $z_h$ times in $W$. Thus there are no occurrences of $e$ in $W'$.
  Thus $e$ can not appear
    in any future cycles, that is  $e \notin C_g$ for any $g > h$. The first observation then follows because
    there are at most $n^2-n$ different edges in $G$. The second observation is then an immediate
    consequence of the first observation, and the invariant established in Lemma \ref{lemma:walkminuscycle} . The third observation follows because
    no edges are ever lost or created in the recursive construction.
\end{proof}

\begin{observation}
\label{obs:highmult}
    There is a cycle $C_s $, $1 \le s \le \ell$,  such that $z_s$ is at least $p+2$.
\end{observation}

\begin{proof}
Since $P$ is a simple path  it contains at most $n-1$ edges.
Thus  $W-P$ contains at least
$n (n^2-n) (p+2) $ edges. As any simple cycle contains at most $n$ edges, and
as there are at most $\ell \le n^2-n$ cycles, then
by applying the pigeon hole principle to   the cycle decomposition of $W$ we can conclude that
there must be a cycle $C_s$ that
has multiplicity at least $p+2$, that is $z_s \ge p+2$.
\end{proof}

For convenience, consider renumbering the cycles so that $s=1$ where $z_s \geq p+1$, which exists  by Observation \ref{obs:highmult}.
For  each $h$ such that $1 \le h \le p+1$, let  define $W_h$ be the collection
of edges in $W$ minus
 $h$ copies of every edge in $C_1$.

\begin{observation}
    \label{obs:Wh}
    The edges in each $W_h$, $1 \le h \le p+1$ form a walk from $i$ to $j$.
\end{observation}

\begin{proof}
As $z_1 \ge p+ 2$ and $h \le p+1$, every edge that appears in $W$ also appears in $W_h$. So
$W_h$ has the same connectivity properties as $W$, and $W_h$ has the same vertices with
positive in-degree as does $W$. Also as $C_1$ is a simple cycle, the difference between
in-degree and out-degree for each vertex is the same in $W_h$ as in $W$.
Thus the result follows by appealing to Lemma \ref{lemma:Eulerian}.
\end{proof}

\begin{observation}
\label{obs:WhPhi}
    For each $h$ such that $1 \le h \le p+1$ it is the case that
    $$\Phi(W_h)= \Phi(P) \Phi(C_1^{z_1 - h}) \prod_{f=2}^\ell  \Phi(C_f^{z_f})$$
\end{observation}
\begin{proof}
This follows directly from the defintion of $W_h$.
\end{proof}

\begin{observation}
\label{observation:shorter}
    For all $h$ such that $1 \le h \le p+1$,
    we have $W_h \in \bigcup_{f =0}^k {\mathcal W^f_{i,j}}$.
    That is $\Phi(W_h)$ appears as a term in $A_{i,j}^{(k)}$.
\end{observation}
\begin{proof}
This follows because $W$ has $k+1$ edges and $C_1$ is non-empty, so removing edges in $C_i$ strictly
decreases the number of edges.
\end{proof}

We are now ready to prove Lemma \ref{lemma:stabilityupper}.
By Observation \ref{observation:shorter} we know that each $\Phi(W_h)$ is included in $A^{(k)}$,
and thus there exists an element $r$ in the semiring $S$ such that
$$A^{(k)}_{i,j} = r+ \sum_{h=1}^{p+1} \Phi(W_h) $$
Thus
\begin{align}
A^{(k)}_{i,j} + \Phi(W) &=  r+ \sum_{h=1}^{p+1} \Phi(W_h) +   \Phi(W) \nonumber\\
&=  r+ \sum_{h=1}^{p+1} \left( \Phi(P) \Phi(C_1^{z_1 - h}) \prod_{f=2}^\ell  \Phi(C_f^{z_f})\right) +   \Phi(P) \prod_{f=1}^\ell  \Phi(C_f^{z_f}) \nonumber\\
&=  r+ \left( \Phi(P)    \prod_{f=2}^\ell  \Phi(C_f^{z_f}) \right)   \left( \sum_{h=1}^{p+1}  \Phi(C_1^{z_1 - h})  +    \Phi(C_1^{z_1}) \right) \label{eq:stability0}\\
&=  r+ \left( \Phi(P)    \prod_{f=2}^\ell  \Phi(C_f^{z_f}) \right)   \left(\Phi(C_1^{z_1 - (p+1)}) \sum_{h=  0}^{p+1}  \Phi(C_1^{ h})   \right) \nonumber\\
&=  r+ \left( \Phi(P)    \prod_{f=2}^\ell  \Phi(C_f^{z_f}) \right)   \left(\Phi(C_1^{z_1 - (p+1)}) \sum_{h=  0}^{p+1}  \left[\Phi(C_1)  \right]^h \right) \label{eq:stability1}\\
&=  r+ \left( \Phi(P)    \prod_{f=2}^\ell  \Phi(C_f^{z_f}) \right)   \left(\Phi(C_1^{z_1 - (p+1)}) \sum_{h=  0}^{p}  \left[ \Phi(C_1)   \right]^h \right)  \label{eq:stability2}\\
&=  r+ \left( \Phi(P)    \prod_{f=2}^\ell  \Phi(C_f^{z_f}) \right)   \left(\Phi(C_1^{z_1 - (p+1)}) \sum_{h=  0}^{p}  \Phi(C_1^{ h})   \right)  \nonumber\\
&=  r+ \left( \Phi(P)    \prod_{f=2}^\ell  \Phi(C_f^{z_f}) \right)   \left( \sum_{h=1}^{p+1}  \Phi(C_1^{z_1 - h})   \right) \nonumber\\
&=  r+ \sum_{h=1}^{p+1} \left( \Phi(P) \Phi(C_1^{z_1 - h}) \prod_{f=2}^\ell  \Phi(C_f^{z_f})\right) \nonumber\\
 &=  r+ \sum_{h=1}^{p+1} \Phi(W_h)   = A^{(k)}_{i,j}  \nonumber
\end{align}
The equality  in line (\ref{eq:stability0}) follows from Observation \ref{obs:WhPhi}.
The equality in line (\ref{eq:stability1}) follows from the defintion of $\Phi$.
The key step in this line of equations is the equality in line (\ref{eq:stability2}), which follows from the stability
of $\Phi(C_1)$. The rest of the equalities  follow from basic algebraic properties of semirings, or by
definition of the relavent term.

\section{A Failed Lower Bound Attempt for  Convergence Rate of Linear $\Name$ Programs}
    \label{sec:pn3lower}

Recall that the definition of $p$-stability is that
$$1 \oplus u \oplus u^2 \oplus \cdots \oplus u^{p}= 1\oplus u \oplus u^2 \oplus \cdots \oplus u^{p}  \oplus u^{p+1}$$
Thus it is natural (but wrong) to conjecture that in order for a walk $W$ to be absorbed by shorter
walks that it must be the case that one can express $\Phi(W)$ as $X  Y^{p+1}$
where $X$ and $Y^{p+1}$ are products of two parts of some partition of $W$.  
In this case then there is at least the possibility that $W$ is absorbed by 
$$X \oplus X \otimes Y \oplus X \otimes Y^2  \ldots \oplus X \otimes Y^p$$
if it happened that there were shorter walks with products equal to these summands. 
We now demonstrate that in order to guarantee that a walk $W$ has the property that 
it can be decomposed so that
 $\Phi(W)= X \otimes Y^{p+1}$ then $W$ has to be of length $\Omega(pn^3)$.

We define a graph $G$. Let $G$ be a directed graph where the vertex set is the integers from 1 to $n$ inclusive. For simplicity assume $n$ is divisible by 3.
The vertices are partitioned into 3 parts, $B=\{1, \ldots, n/3\}$, $C=\{n/3 +1, \ldots, 2n/3\}$
and $D=\{2n/3+1, \ldots n\}$.
There is a directed edge from each vertex in $B$ to vertex $n/3+1$,
there is a directed edge from vertex $2n/3$ to each vertex in $D$,
and there is a directed edge from each vertex in $D$ to each vertex in $B$. Finally, all vertices
in $C$ are sequentially connected from $n / 3 +1$ to $2n / 3$, i.e., there is a directed edge from $\tau$ to $\tau +1$ for all $\tau \in [n/ 3 + 1, 2n / 3 - 1]$. See Figure \ref{fig:G}.

\begin{figure}[h]
  \vspace{-2mm}
  \centering
  \includegraphics[width=0.7\textwidth]{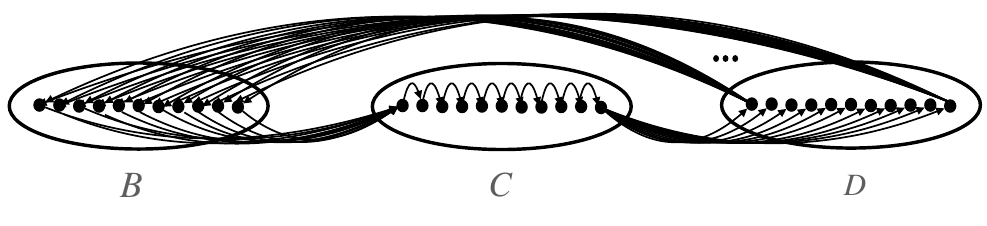}
  \caption{The graph $G$.}
  \label{fig:G}
\end{figure}

%Note that there are $|B| \cdot |D| = n^2 / 9$ edges from $D$ to $B$. Consider a long walk, say from $i := n/3 + 1$ to $j := n/3 + 2$. Observe that the walk must visit all edges within $C$ (and the edge from $2n/3$ to $2n/3 +1$) before visiting exactly one edge from $D$ to $B$. Thus, to visit each edge from $D$ to $B$ at least $p$ times, the walk must have length at least $p |B| \cdot |D| \cdot |C| = p n^3  / 27$. 

Note that there are $|B| \cdot |D| = n^2 / 9$ edges from $D$ to $B$. Consider a long walk, say from $i := n/3 + 1$ to $j := n/3 + 2$, that visits each edge from $D$ to $B$ exactly $p$ times. Clearly, for such a walk $W$, we cannot express $W$  as $X  Y^{p+1}$ for some $X$ and $Y$.
 Observe that the walk must visit all edges within $C$ (and the edge from $2n/3$ to $2n/3 +1$) before visiting exactly one edge from $D$ to $B$. 
Thus, to visit each edge from $D$ to $B$ at least $p$ times, the walk must have length at least $p |B| \cdot |D| \cdot |C| = p n^3  / 27$.

\section{Bounding Convergence in Terms of the Semiring Ground Set Size}
\label{sect:Lupper}

This section investigates the convergence rate with the assumption that the semiring has a ground set of size at most $L$.  With this assumption, we can prove significantly better upper bounds.  This section's goal is to prove Theorem~\ref{thm:plupper}.
We prove the following lemma, which will immediately imply Theorem~\ref{thm:plupper}.
As before, we let $\Phi(W) = \prod_{e \in W} e$ for a walk $W$.

\begin{lemma}
        Let $A$ be an $n$ by $n$ matrix
    both a $p$-stable semiring $S$ with a ground set consisting of $L$ elements.
    Then $A^{(k+1)} = A^{(k)}$, where $k= \lceil 8p  (\lg L +1)n \rceil + 1= O(n p \lg L)$.     
\end{lemma}
\begin{proof}
      Fix $i, j \in [n]$.
    Consider any $W \in \calW_{i,j}^{k+1}$ for the value of $k$ stated in the lemma. To show the lemma it suffices to show
    $\Phi(W) + A_{i,j}^{(k)} = A_{i,j}^{(k)}$. By the pigeon hole principle, there must exist a vertex $v$ visited at least $8p  (\lg L +1)+1$ times. We now conceptually cut $W$ at visits to $v$ to form
    a cycle decomposition of $W$.
    That is, we can write $W$ as $T C_1, \ldots, C_h T'$ where $T$ is a walk from $i$ to $v$
    such that the only time it visits vertex $v$ is on the last step, each $C_i$ is a closed walk that
    includes vertex $v$ exactly once, and $T'$ is a walk from $v$ to $j$ that doesn't visit $v$ again after initially
    leaving $v$. Note by the definition of vertex $v$  it must be the case that
    $h\ge 8p  (\lg L +1)$. Let $\cC = \{C_f \mid 1 \le f \le h \}$ be the
    collection of cycles in this cycle decomposition.

We now partition $\cC$ into parts ${\cC}_1, \ldots, {\cC}_\ell, {\mathcal J}$
where for each closed walk $C $ that has multiplicity $m$ in $\cC$ there are
 $2^f$ copies of  $C$  in $\cC_f$ and $m - 2^f$ copies of $C$ in $\mathcal J$ where
 $f = \lfloor \lg m \rfloor$.
For a collection $\mathcal D$ of cycles, it will be convenient to use
$\Phi({\mathcal D})$ to denote $ \Phi(\bigcup {\mathcal D})$.
Thus it then immediate that $\Phi(W) = \Phi(T) \Phi(\cC) \Phi(T') \Phi({\mathcal J})$.
Note that the cardinality of the multiset $\bigcup_{f=1}^\ell \cC_f$ is at least $4p(\lg L + 1)$.

    We   change $\cC$ repeatedly without changing $\Phi(\cC)$. When changing $\cC$ into $\cC'$, we satisfy:
    \begin{enumerate}
        \item $\Phi(\cC) = \Phi(\cC')$.
        \item $N(\cC) = N(\cC')$, where $N(\cC)$ denotes the size of multi-set $\cC$. So, a cycle $C$ contributes to $N(\cC)$ by the number of times it appears in $\cC$. Further, $\cC'$ has no more edges in total than $\cC$ when we count 1 for each edge in one cycle appearance, i.e., $\sum_{C' \in \cC'} |C'| \leq \sum_{C \in \cC} |C|$.

        \item Consider $\langle \ldots , N_3(\cC'), N_2(\cC'), N_1(\cC') \rangle$ and $\langle \ldots , N_3(\cC), N_2(\cC), N_1(\cC) \rangle$. The first vector dominates the second lexicographically. Here $N_\ell(\cC)$ denotes the number of cycles in $\cC$ of exponent $2^\ell$.
        \item When we terminate, $N_\ell(\cC) \leq 2 \lfloor \lg L +1 \rfloor$ for all $\ell \leq  \lfloor \lg p \rfloor$.
        \item Every cycle in $\cC'$ also appears in $\cC$.
    \end{enumerate}

    We now describe the transformation from $\cC$ into $\cC'$. Consider the smallest $\ell$ such that $N_\ell > 2 \lfloor \lg L+1 \rfloor$. Consider every subset  of cardinality $\lfloor \lg L+1 \rfloor$, that consists of cycles of exponent $2^\ell$; so they are in $\cC_\ell$. The number of such subsets is at least ${2 \lfloor \lg L+1 \rfloor \choose  \lfloor \lg L+1 \rfloor} > 2^{\lg L} = L$. Since the semiring has at most $L$ distinct elements, there must exist distinct subsets $A$ and $B$ of cycles  of exponent $2^\ell$ such that $|A| = |B|$ and $\Phi(A) = \Phi(B)$. Assume wlog that
    $B$'s cycles have no more edges in total than $A$'s cycles. Then, we replace $A$ with $B$  in $\cC$  and let $\cC'$ be $\cC$ after this change.

    It is easy to see that $\Phi(\cC) = \Phi(\cC')$. This is because we replaced $A$ with $B$ such that $\Phi(A)= \Phi(B)$. More formally, $\Phi(\cC) = \prod_{\ell'} \Phi(\cC_{\ell'}) =
    (\Phi(\cC_{\ell} \setminus A)\Phi(A)) \prod_{\ell' \neq \ell} \Phi(\cC_{\ell'})
    = (\Phi(\cC_{\ell} \setminus A)\Phi(B)) \prod_{\ell' \neq \ell} \Phi(\cC_{\ell'}) = \Phi(\cC')$.
    The second property is also obvious because when we replace $A$ with $B$ in the change, we ensured $|A| = |B|$, which implies that the multiset size remains unchanged. Also, it is immediate that the total number of edges don't increase as $B$ has no more edges in total than $A$.
    The forth property, the termination condition, is immediate. The fifth property is also immediate since we do not create a new cycle when replacing $A$ with $B$.

    To see the third property, before  replacing $A$ with $B$, we had $A \cup B$ in $\cC_\ell$, and their exponent was $2^\ell$. After the replacement, all cycles in $B \setminus A$ come to have exponent $2 \cdot 2^\ell = 2^{\ell+1}$, the cycles in $A \setminus B$ disappear from $\cC_\ell$, and those in $A \cap B$ remain unchanged. Note that every cycle remains to have an exponent that is a power of two.
     Therefore, we have $N_{\ell+1}(\cC') > N_{\ell+1}(\cC)$, and $N_{\ell'}(\cC') = N_{\ell'}(\cC)$  for all $\ell' \geq \ell+2$.

    Observe that due to the second property and the third, the process must terminate.
    Thus, starting from $\cC$, at the termination we have $\cC'$ that satisfies the first, second, and fourth properties. We know $N(\cC') = N(\cC) > 4p  (\lg L +1)$. Further, we know that cycles of exponent at most $p$ contribute to $N(\cC')$ by at most $2(\lg L +1) (1+ 2 + 4 + \ldots + 2^{\lfloor\lg p \rfloor}) < 4 p (\lg L +1) $. Thus, there must exist a cycle in $\cC'$ of exponent greater than $p$. Let $C$ denote the cycle. 
    So, we have shown that
    $$\Phi(W) = \Phi(T) \Phi(\cC') \Phi(T')= \Phi(T) \Phi(C)^q \Phi(\cC' \setminus C^{q}) \Phi(T'),$$
    where $q \geq p+1$. Here $\cC' \setminus C^q$ implies the resulting collection of cycles we obtain after removing $q$ copies of $C$ from $\cC'$. Now consider walk $W_{q'}$ that concatenates $T$, $q'$ copies of $C$, $\cC' \setminus C^{q'}$, and $T'$. Here, the walk starts with $T$ and ends with $T'$, and the cycles can be placed in an arbitrary order. This is because
    $T$ is a walk from $i$ to $v$, and all the cycles in $\cC'$ start from $v$ and end at $v$---due to the fifth property---and $T'$ is a walk from $v$ to $j$
    Further, they are all shorter than $W$  because $W_q$ is no longer than $W$ due to the second property. Therefore, $W_0, W_1, \ldots W_{q-1} \in  \calW_{i,j}^{(k)}$. Thus, thanks to $p$-stability, we have
    $\Phi(W) + A_{i,j}^{(k)} = A_{i,j}^{(k)}$ as desired.        
\end{proof}

Although we gave the full proof of Theorem~\ref{thm:plupper}, to convey intuition better, we also give
some warm-up analyses in Appendix~\ref{seca:Lupper} by giving a looser bound for the general case and subsequently by considering a special case of $p = 1$.

\section{Related Work}\label{sec:related}

If the semiring is naturally ordered\footnote{
$\bm S$ is naturally ordered if the relation $x \preceq_S y$ defined as
$\exists z: x \oplus z = y$ is a partial order.
},
then the
least fixpoint of a $\Name$ program is the least fixed point of $f$ under the same partial
order extended to $S^n$ componentwise. This is the {\em least fixpoint} semantics of a
$\Name$ program. The na\"ive evaluation algorithm for evaluating
$\Datalog$ programs extends naturally to evaluating $\Name$ programs: starting from $\bm
x=0^n$, we repeatedly apply $f$ to $\bm x$ until a fixpoint is reached $\bm x = f(\bm x)$.
The core semiring of a POPS is naturally ordered.
Thus, we can find the least fixpoint of a $\Name$ program by applying the na\"ive
evaluation algorithm~\cite{Khamis0PSW22}.

Computing the least fixpoint solution to a recursive $\Name$ program boils down to solving
fixpoint equations over semirings. In particular, we are given a multi-valued polynomial
function $f : S^n \to S^n$ over a commutative semiring, and the problem is to compute a
(pre-) fixpoint of $f$, i.e. a point $x \in S^n$ where $x=f(x)$. As surveyed in 
in~\cite{Khamis0PSW22}, this problem was studied in a very wide range of communities, such
as in automata theory~\cite{MR1470001}, program
analysis~\cite{DBLP:conf/popl/CousotC77,MR1728440}, and  graph
algorithms~\cite{MR556411,MR584516,MR526496} since the 1970s.
(See~\cite{MR1059930, DBLP:conf/lics/HopkinsK99, DBLP:journals/tcs/Lehmann77,
semiring_book,MR609751} and references thereof).

When $f = Ax+b$ is linear, as shown in the paper $f^{(k)}(x) = A^{(k-1)}b$ and thus at
fixpoint the solution is $A^{(\omega)}b = \lim_{k\to\infty} A^{(k-1)}b$, interpreted as a
formal power series over the semiring. If there is a finite $k$ for which $A^{(k)} =
A^{(k+1)}$, then it is easy to see that $A^{(\omega)} = A^{(k)}$. The problem of computing
$A^{(\omega)}$ is called the {\em algebraic path problem}~\cite{MR1059930}, which unifies
many problems such as transitive closure~\cite{DBLP:journals/jacm/Warshall62}, shortest
paths~\cite{DBLP:journals/cacm/Floyd62a}, Kleene's theorem on finite automata and regular
languages~\cite{MR0077478},  and continuous
dataflow~\cite{DBLP:conf/popl/CousotC77,DBLP:journals/jacm/KamU76}. If $A$ is a real
matrix, then $A^{(\omega)} = I+A+A^2+\cdots$ is exactly $(I-A)^{-1}$, if it
exists~\cite{MR427338,tarjan1,MR371724}.

There are several classes of solutions to the algebraic path problem, which have pros and
cons depending on what we can assume about the underlying semiring (whether or not there is
a closure operator, idempotency, natural orderability, etc.). We refer the reader
to~\cite{semiring_book,MR1059930} for more detailed discussions. 

\bibliography{main}

\begin{thebibliography}{10}

\bibitem{DBLP:books/aw/AbiteboulHV95}
Serge Abiteboul, Richard Hull, and Victor Vianu.
\newblock {\em Foundations of Databases}.
\newblock Addison-Wesley, 1995.
\newblock URL: \url{http://webdam.inria.fr/Alice/}.

\bibitem{MR427338}
R.~C. Backhouse and B.~A. Carr\'{e}.
\newblock Regular algebra applied to path-finding problems.
\newblock {\em J. Inst. Math. Appl.}, 15:161--186, 1975.

\bibitem{MR556411}
Bernard Carr\'{e}.
\newblock {\em Graphs and networks}.
\newblock The Clarendon Press, Oxford University Press, New York, 1979.
\newblock Oxford Applied Mathematics and Computing Science Series.

\bibitem{DBLP:conf/popl/CousotC77}
Patrick Cousot and Radhia Cousot.
\newblock Abstract interpretation: {A} unified lattice model for static
  analysis of programs by construction or approximation of fixpoints.
\newblock In Robert~M. Graham, Michael~A. Harrison, and Ravi Sethi, editors,
  {\em Conference Record of the Fourth {ACM} Symposium on Principles of
  Programming Languages, Los Angeles, California, USA, January 1977}, pages
  238--252. {ACM}, 1977.
\newblock \href {https://doi.org/10.1145/512950.512973}
  {\path{doi:10.1145/512950.512973}}.

\bibitem{DBLP:journals/cacm/Floyd62a}
Robert~W. Floyd.
\newblock Algorithm 97: Shortest path.
\newblock {\em Commun. {ACM}}, 5(6):345, 1962.
\newblock \href {https://doi.org/10.1145/367766.368168}
  {\path{doi:10.1145/367766.368168}}.

\bibitem{ParisPersonal}
Simon Frisk, Paris Koutris, and Hangdong Zhao.
\newblock personal communication, 2025.

\bibitem{MR371724}
M.~Gondran.
\newblock Alg\`ebre lin\'{e}aire et cheminement dans un graphe.
\newblock {\em Rev. Fran\c{c}aise Automat. Informat. Recherche
  Op\'{e}rationnelle S\'{e}r. Verte}, 9({\rm V}-1):77--99, 1975.

\bibitem{semiring_book}
Michel Gondran and Michel Minoux.
\newblock {\em Graphs, dioids and semirings}, volume~41 of {\em Operations
  Research/Computer Science Interfaces Series}.
\newblock Springer, New York, 2008.
\newblock New models and algorithms.

\bibitem{DBLP:conf/lics/HopkinsK99}
Mark~W. Hopkins and Dexter Kozen.
\newblock Parikh's theorem in commutative kleene algebra.
\newblock In {\em 14th Annual {IEEE} Symposium on Logic in Computer Science,
  Trento, Italy, July 2-5, 1999}, pages 394--401. {IEEE} Computer Society,
  1999.
\newblock \href {https://doi.org/10.1109/LICS.1999.782634}
  {\path{doi:10.1109/LICS.1999.782634}}.

\bibitem{icdt/ImM0P24}
Sungjin Im, Benjamin Moseley, Hung~Q. Ngo, and Kirk Pruhs.
\newblock On the convergence rate of linear datalog {\^{}}{\(\circ\)} over
  stable semirings.
\newblock In Graham Cormode and Michael Shekelyan, editors, {\em 27th
  International Conference on Database Theory, {ICDT} 2024, March 25-28, 2024,
  Paestum, Italy}, volume 290 of {\em LIPIcs}, pages 11:1--11:20. Schloss
  Dagstuhl - Leibniz-Zentrum f{\"{u}}r Informatik, 2024.
\newblock URL: \url{https://doi.org/10.4230/LIPIcs.ICDT.2024.11}, \href
  {https://doi.org/10.4230/LIPICS.ICDT.2024.11}
  {\path{doi:10.4230/LIPICS.ICDT.2024.11}}.

\bibitem{DBLP:journals/jacm/KamU76}
John~B. Kam and Jeffrey~D. Ullman.
\newblock Global data flow analysis and iterative algorithms.
\newblock {\em J. {ACM}}, 23(1):158--171, 1976.
\newblock \href {https://doi.org/10.1145/321921.321938}
  {\path{doi:10.1145/321921.321938}}.

\bibitem{Khamis0PSW22}
Mahmoud~Abo Khamis, Hung~Q. Ngo, Reinhard Pichler, Dan Suciu, and Yisu~Remy
  Wang.
\newblock Convergence of datalog over (pre-) semirings.
\newblock In Leonid Libkin and Pablo Barcel{\'{o}}, editors, {\em {PODS} '22:
  International Conference on Management of Data, Philadelphia, PA, USA, June
  12 - 17, 2022}, pages 105--117. {ACM}, 2022.
\newblock \href {https://doi.org/10.1145/3517804.3524140}
  {\path{doi:10.1145/3517804.3524140}}.

\bibitem{MR0077478}
S.~C. Kleene.
\newblock Representation of events in nerve nets and finite automata.
\newblock In {\em Automata studies}, Annals of mathematics studies, no. 34,
  pages 3--41. Princeton University Press, Princeton, N. J., 1956.

\bibitem{MR1470001}
Werner Kuich.
\newblock Semirings and formal power series: their relevance to formal
  languages and automata.
\newblock In {\em Handbook of formal languages, {V}ol. 1}, pages 609--677.
  Springer, Berlin, 1997.

\bibitem{DBLP:journals/tcs/Lehmann77}
Daniel~J. Lehmann.
\newblock Algebraic structures for transitive closure.
\newblock {\em Theor. Comput. Sci.}, 4(1):59--76, 1977.
\newblock \href {https://doi.org/10.1016/0304-3975(77)90056-1}
  {\path{doi:10.1016/0304-3975(77)90056-1}}.

\bibitem{MR526496}
Richard~J. Lipton, Donald~J. Rose, and Robert~Endre Tarjan.
\newblock Generalized nested dissection.
\newblock {\em SIAM J. Numer. Anal.}, 16(2):346--358, 1979.
\newblock URL: \url{https://doi-org.gate.lib.buffalo.edu/10.1137/0716027},
  \href {https://doi.org/10.1137/0716027} {\path{doi:10.1137/0716027}}.

\bibitem{MR584516}
Richard~J. Lipton and Robert~Endre Tarjan.
\newblock Applications of a planar separator theorem.
\newblock {\em SIAM J. Comput.}, 9(3):615--627, 1980.
\newblock URL: \url{https://doi-org.gate.lib.buffalo.edu/10.1137/0209046},
  \href {https://doi.org/10.1137/0209046} {\path{doi:10.1137/0209046}}.

\bibitem{MR1728440}
Flemming Nielson, Hanne~Riis Nielson, and Chris Hankin.
\newblock {\em Principles of program analysis}.
\newblock Springer-Verlag, Berlin, 1999.
\newblock URL:
  \url{https://doi-org.gate.lib.buffalo.edu/10.1007/978-3-662-03811-6}, \href
  {https://doi.org/10.1007/978-3-662-03811-6}
  {\path{doi:10.1007/978-3-662-03811-6}}.

\bibitem{MR1059930}
G\"{u}nter Rote.
\newblock Path problems in graphs.
\newblock In {\em Computational graph theory}, volume~7 of {\em Comput.
  Suppl.}, pages 155--189. Springer, Vienna, 1990.
\newblock URL:
  \url{https://doi-org.gate.lib.buffalo.edu/10.1007/978-3-7091-9076-0_9}, \href
  {https://doi.org/10.1007/978-3-7091-9076-0\_9}
  {\path{doi:10.1007/978-3-7091-9076-0\_9}}.

\bibitem{tarjan1}
Robert~E. Tarjan.
\newblock Graph theory and gaussian elimination, 1976.
\newblock J.R. Bunch and D.J. Rose, eds.

\bibitem{DBLP:journals/jacm/Warshall62}
Stephen Warshall.
\newblock A theorem on boolean matrices.
\newblock {\em J. {ACM}}, 9(1):11--12, 1962.
\newblock \href {https://doi.org/10.1145/321105.321107}
  {\path{doi:10.1145/321105.321107}}.

\bibitem{wilson10}
Robin~J. Wilson.
\newblock {\em Introduction to Graph Theory}.
\newblock Prentice Hall/Pearson, New York, 2010.

\bibitem{MR609751}
U.~Zimmermann.
\newblock Linear and combinatorial optimization in ordered algebraic
  structures.
\newblock {\em Ann. Discrete Math.}, 10:viii+380, 1981.

\end{thebibliography}

\appendix

\section{Convergence of Naturally Ordered Semirings}

\begin{definition}[Natural Order]
In any (pre-)semiring $\bm S$, the relation $x \preceq_S y$ defined as
$\exists z: x \oplus z = y$, is a {\em preorder}, which means that it
is reflexive and transitive, but it is not anti-symmetric in general.
When $\preceq_S$ is anti-symmetric, then it is a partial order, and is
called the {\em natural order} on $\bm S$; in that case we say that
$\bm S$ is {\em naturally ordered}.
\end{definition}

For simplicity, we may use $\leq $ in lieu of $\preceq$. We naturally extend the ordering to vectors and matrices: for two vectors $\bm v, \bm w  \in \bm S^n$, we have $\bm v \leq \bm w$ iff $\bm v_i \leq \bm w_i$ for all $i \in [n]$. Similarly, for two matrices $A$ and $B$ (which includes vectors) $A \le B$ means
that componentwise each entry in $A$ is at most the entry in $B$,
that is for all $i$ and $j$ it is the case that $A_{i,j} \le B_{i,j}$.

Here we take $k$ to be the longest chain in the natural order.

 \begin{theorem}
  Every linear $\Name$ program over a $p$-stable naturally-ordered commutative semiring with maximum chain size $k$ converges
  in $O(k n )$ steps.
  \end{theorem}

    \begin{theorem}
  There are linear $\Name$ programs over a $p$-stable naturally-ordered commutative semiring with maximum chain size $k$  that require
  in $\Omega(k n )$ steps to converge.
\end{theorem}

This section considers bounds in terms of the longest chain in the partial order of a naturally ordered semiring.  Recall that natural ordering means the following: for two elements $a$ and $b$  $a \le b $ if and only if there exists a $c$ such that $a+c = b$. Let $L$ be the length of the longest chain in this partial order. We seek to bound the convergence rate in terms of $n$ and $L$.

\begin{lemma}
    Consider a naturally ordered semiring where $L$ is the length of the longest chain in the partial order. Let $A$ be an $n \times n$ matrix.     Convergence must occur within $nL$ steps.
\end{lemma}

\begin{proof}
   Consider $A^{(k)} x $ as $k$ increases for any fixed $x$. If there is a $k \leq nL$ such that
 $A^{(k)} x = A^{(k+1)} x $ then convergence has been reached within the desired number of steps. Otherwise when  $A^{(k)} x \neq A^{(k+1)} x $
 there exists an $i$ such that dimension $i$ in $A^{(k+1)} x $ is strictly greater than dimension $i$ in $A^{(k)} x $.
This can only occur $L$ times for each $i$ by definition of the partial order. Knowing that there are at most $n$ dimensions in  $A^{(k)} x$, the lemma follows.
\end{proof}

\begin{lemma}
    There exists a naturally ordered semiring  where $L$ is the length of the longest chain in the partial order and a $n$ by $n$ matrix where convergence requires
  $\Omega(nL)$ steps.
\end{lemma}
\begin{proof}
Consider the following semiring.  The semiring is on the set of integers $0,1,2, \ldots, L$ and a special element $\zero$. Here, the additive identity is $\zero$ and the multiplicative identity is 0. Consider two elements $a$ and $b$ that are not $\zero$. Define the addition and multiplication of $a$ and $b$ to be equal to $\min\{a+b, L\}$.  Define $a$ multiplied by $\zero$ to be $\zero$ for any $a$ and $a$ added to $\zero$ to be $a$ for any $a$.  Intuitively, addition and multiplication act as standard addition capped at $L$, except for the special $\zero$ element.

Consider the following graph corresponding to a $n \times n$ matrix $A$.  There is a cycle on $n$ nodes.  Order the edges from $1$ to $n$.  Each edge is labeled $0$ except the edge from $1$ to $2$, which is labeled as $1$. Traversing the cycle $k$ times and multiplying the labels of the edges returns $\min \{k,L\}$.  It takes a walk of length $\Omega(nL)$ to reach the element $L$.
\end{proof}

\section{Warm-up for Proof of Theorem~\ref{thm:plupper}}
\label{seca:Lupper}

In Section~\ref{sect:Lupper} we gave the full proof of Theorem~\ref{thm:plupper}, which gives an upper bound of $O(n p \log L)$ on the convergence rate  when the underlying semiring has a ground set of size at most $L$.

To convey intuition better of the analysis, we give two warm-up proofs. Our first warm-up is giving a looser bound on the convergence rate. The proof makes use of the fact that a sufficiently long walk must visit the same vertex many times with the same product value. Here, we think of a prefix of the walk as a product of edges on the prefix.  This prefex evaluates to an element of the semiring.

\begin{lemma}
    \label{lem:bounded-semiring-loose}
    Let $A$ be an $n$ by $n$ matrix
    over a $p$-stable semiring on a ground set  $S$ consisting of $L$ elements.
    Then $A^{(k+1)} = A^{(k)}$, where $k= n p L$.
\end{lemma}
\begin{proof}
    Fix $i, j \in [n]$.
    Consider any $W \in \calW_{i,j}^{k+1}$. To show the lemma it suffices to show
    $\Phi(W) + A_{i,j}^{(k)} = A_{i,j}^{(k)}$.
    Let $W_h$ be the prefix of $W$ of length $h$. Let $v(W_h)$ denote the ending point of $W_h$. Consider all pairs $(\Phi(W_h), v(W_h))$, $h \in [k+1]$. Since these tuples are subsets of $S \times [n]$ and $|S| = L$, due to the pigeonhole principle, there must exist $H \subseteq [k+1]$ of size $p+1$ such that $(\Phi(W_h), v(W_h))$ is the same tuple for all $h \in H$. By renaming we can represent the prefixes as $X_1$, $X_1 X_2$, $X_1 X_2 X_3$, $\ldots$, $X_1 X_2 X_3 \ldots X_{p+1}$.

    For some $T$ (possibly empty), we have $W = X_1 X_2 X_3 \ldots X_{p+1} T$. By definition  $\Phi(X_1) = \Phi(X_1 X_2) = \ldots = \Phi(X_1 X_2 \ldots X_{p+1})$.
    
    Thus, %   \begin{align*}
        $\Phi(W) = \Phi(X_1 X_2 X_3 \ldots X_{p+1} T)$   
               $ = \Phi(X_1 X_2 X_3 \ldots X_{p} T) $  
              $  \ldots 
                = \Phi(X_1  T)$, This uses the fact that $X_1  T$, $X_1 X_2 T$, $\ldots $, $X_1 X_2 X_3 \ldots X_{p} T$ all are walks from $i$ to $j$ since $X_1$, $X_1 X_2$, \ldots, $X_1 X_2 \ldots X_{p+1}$ all end where $T$ starts.
    Further, all walks $X_1  T$, $X_1 X_2 T$, $\ldots $, $X_1 X_2 X_3 \ldots X_{p} T$ are strictly shorter than $W$. 
    This implies that  $\Phi(W)$ appears at least $p$ times in $A^{(k)}_{i,j}$. Using Proposition~\ref{prop:coef-limit}, we conclude $\Phi(W) + A_{i,j}^{(k)} = A_{i,j}^{(k)}$ as desired.
\end{proof}

Next we consider the special case of $p = 1$. In this case we give an exponential improvement over what we showed in the previous lemma. The key idea is the following. Previously we identified $p$ disjoint cycles $X_2, X_3, \ldots, X_{p+1}$ that share the same starting and ending vertex from a long walk $W$  in $\calW^{k+1}_{i, j}$. Then, by removing them sequentially we were able to obtain $p$ copies of the same element that have already appeared; thus adding $W$ (or more precisely $\Phi(W)$) doesn't change $A^{(k)}_{i, j}$. Now we would like to make the same argument with an exponentially smaller number of cycles. Roughly speaking, we will identify $\Theta(\lg L)$ such cycles and find $2^{\Theta(\lg L)}$ walks by combining subsets of them. That is,  the key idea is that  we find more walks with  the same product from far fewer cycles.

\begin{lemma}       
        Let $A$ be an $n$ by $n$ matrix
    both over a $1$-stable semiring $S$ with a ground set consisting of $L$ elements.
    Then $A^{(k+1)} = A^{(k)}$, where $k= O(n  \lg L)$.
\end{lemma}
\begin{proof}    
    As before,  fix $i, j \in [n]$. Consider any  $k \geq \lceil 2 \lg L \rceil n$.
    For any $W \in \calW_{i,j}^{k+1}$ we show
    $\Phi(W) + A_{i,j}^{(k)} = A_{i,j}^{(k)}$.
    Since there are $n$ vertices, the walk must visit some vertex at least $\lceil 2  \lg L \rceil + 1$ times. Formally, we can decompose $W$ into
    \begin{equation}
        \label{eqn:warm-decom}
            W = T C_1 C_2 \ldots C_H T'
    \end{equation}
    where $T$, $T C_1$, $T C_1 C_2$, \ldots,  $T C_1 C_2 \ldots C_H$ all end at the same vertex $v$, and $H = \lceil 2 \lg L \rceil$. Note that all the cycles (or closed walks) $C_1, C_2, \ldots C_H$  start from $v$ and end at the same vertex $v$. It is plausible that some of them are identical.

     For a subset $A$  of $[H]$ we let  $\hat \Phi(A) :=  \prod_{h \in A} \Phi(C_h)$.
     Since there are $2^{|H|}$ subsets of $[H]$ and $2^H > L$, there must exist $A, B \subseteq [H]$ such that $A\neq B$ and  $\hat \Phi(A) = \hat \Phi(B)$. Assume wlog that $B \setminus  A \neq \emptyset$.
Thus we know
\begin{equation}
    \label{eqn:200}
\hat \Phi(A \cap B) \hat \Phi(A \setminus B) = \hat \Phi(A \cap B) \hat \Phi(B \setminus A)
\end{equation}
We can then show,
\begin{align*}
    \Phi(W) &= \Phi(T) \hat \Phi([H]) \Phi(T') \quad \quad \mbox{[Eqn.~\ref{eqn:warm-decom}]}\\
    & = \Phi(T) \hat \Phi(A \cap B) \hat \Phi(A \setminus B)  \hat \Phi(B \setminus A) \hat \Phi( [H]\setminus (A \cup B)) \Phi(T') \\
    & = \Phi(T) \hat \Phi(A \cap B) (\hat \Phi(B \setminus A))^2 \hat \Phi( [H]\setminus (A \cup B)) \Phi(T') \quad \quad \mbox{[Eqn.~\ref{eqn:200}]}
\end{align*}
Consider a walk $W'$ that starts with $T$, has $C_h$ for each $h \in (A \cap B) \cup (B \setminus A) \cup ([H] \setminus (A \cup B)) = [H] \setminus (A \setminus B)$ and ends with $T'$. Similarly, consider a walk $W''$ that starts with $T$, has $C_h$ for each $h \in (A \cap B) \cup  ([H] \setminus (A \cup B))$ and ends with $T'$. Note that $W$ and $W'$ are different since $B \setminus A \neq \emptyset$. Further they  are walks from $i$ to $j$ since every cycle $C_h$, $h \in [H]$ starts from and ends at the same vertex $v$. Further, both walks are shorter than $W$, and therefore are in $\calW^{(k)}_{i,j}$. Since we have
$  \Phi(W') = \Phi(T) \hat \Phi(A \cap B) \hat \Phi(B \setminus A) \hat \Phi( [H]\setminus (A \cup B)) \Phi(T')$ and $        \Phi(W'') = \Phi(T) \hat \Phi(A \cap B)  \hat \Phi( [H]\setminus (A \cup B)) \Phi(T') $.  We will show using 1-stability of the semiring that  $\Phi(W') + \Phi(W'') + \Phi(W) = \Phi(W') + \Phi(W'')$, implying
$\Phi(W) + A_{i,j}^{(k)} = A_{i,j}^{(k)}$ as desired.

Thus, it suffices to show  $\Phi(W') + \Phi(W'') + \Phi(W) = \Phi(W') + \Phi(W'')$. To see this:
\begin{align*} &\Phi(W') + \Phi(W'') + \Phi(W)\\
&=  \Phi(T) \hat \Phi(A \cap B) \hat \Phi(B \setminus A) \hat \Phi( [H]\setminus (A \cup B)) \Phi(T')  + \Phi(T) \hat \Phi(A \cap B)  \hat \Phi( [H]\setminus (A \cup B)) \Phi(T') \\
&\qquad+\Phi(T) \hat \Phi(A \cap B) (\hat \Phi(B \setminus A))^2 \hat \Phi( [H]\setminus (A \cup B)) \Phi(T')   &\\
&= \Phi(T) \hat \Phi(A \cap B) \Phi( [H]\setminus (A \cup B)) \Phi(T')  (1 + \Phi(B \setminus A)+\Phi(B \setminus A)^2 )\\
&\qquad [\mbox{associative and 1 is the multiplicative identiy}]\\
&= \Phi(T) \hat \Phi(A \cap B) \Phi( [H]\setminus (A \cup B)) \Phi(T')  (1 + \Phi(B \setminus A)) \qquad [\mbox{1-stable}]\\
&= \Phi(W') + \Phi(W'')
\end{align*}
    \end{proof}

\end{document}